\documentclass[11pt]{article}

\usepackage[T1]{fontenc}
\usepackage{lmodern}
\usepackage{microtype}
\usepackage[margin=1in]{geometry}
\usepackage{amsmath,amssymb,amsthm,mathtools}
\numberwithin{equation}{section}
\usepackage{bbm}
\usepackage{aliascnt}
\usepackage{enumitem}
\usepackage{xcolor}
\usepackage[normalem]{ulem}
\usepackage{url}
\usepackage[colorlinks=true,linkcolor=blue!48!black,citecolor=green!35!black,urlcolor=blue!55!black]{hyperref}

\hypersetup{
  pdftitle={Improved Lower Bounds for Decomposable Randomized Encoding},
  pdfauthor={Justin Holmgren and Kewen Wu},
  pdfsubject={Lower bounds for decomposable randomized encodings of OR and general Boolean functions}
}

\setlist[itemize]{leftmargin=1.7em,itemsep=0.2em,topsep=0.35em}
\setlist[enumerate]{leftmargin=1.9em,itemsep=0.2em,topsep=0.35em}
\allowdisplaybreaks

\newaliascnt{lemma}{theorem}
\newtheorem{lemma}[lemma]{Lemma}
\aliascntresetthe{lemma}
\newaliascnt{proposition}{theorem}

\aliascntresetthe{proposition}
\newaliascnt{corollary}{theorem}
\newtheorem{corollary}[corollary]{Corollary}
\aliascntresetthe{corollary}
\newaliascnt{claim}{theorem}

\aliascntresetthe{claim}
\newaliascnt{conjecture}{theorem}
\newtheorem{conjecture}[conjecture]{Conjecture}
\aliascntresetthe{conjecture}
\theoremstyle{definition}
\newaliascnt{definition}{theorem}
\newtheorem{definition}[definition]{Definition}
\aliascntresetthe{definition}
\newaliascnt{construction}{theorem}

\aliascntresetthe{construction}
\newaliascnt{example}{theorem}

\aliascntresetthe{example}
\theoremstyle{remark}
\newaliascnt{remark}{theorem}
\newtheorem{remark}[remark]{Remark}
\aliascntresetthe{remark}
\newaliascnt{question}{theorem}

\aliascntresetthe{question}
\usepackage{thm-restate}
\usepackage[capitalise,nameinlink,noabbrev]{cleveref}
\crefname{definition}{definition}{definitions}
\Crefname{definition}{Definition}{Definitions}
\crefname{construction}{construction}{constructions}
\Crefname{construction}{Construction}{Constructions}
\crefname{claim}{claim}{claims}
\Crefname{claim}{Claim}{Claims}
\crefname{conjecture}{conjecture}{conjectures}
\Crefname{conjecture}{Conjecture}{Conjectures}
\crefname{example}{example}{examples}
\Crefname{example}{Example}{Examples}
\crefname{question}{open question}{open questions}
\Crefname{question}{Open Question}{Open Questions}

\DeclareMathOperator*{\E}{\mathbb E}
\renewcommand{\Pr}{\operatorname*{\mathbf{Pr}}}
\newcommand{\lcm}{\operatorname{lcm}}

\newcommand{\vabs}[1]{\left\| #1 \right\|}

\newcommand{\pbra}[1]{\left( #1 \right)}
\newcommand{\sbra}[1]{\left[ #1 \right]}
\newcommand{\cbra}[1]{\left\{ #1 \right\}}

\newcommand{\bin}{\{0,1\}}
\newcommand{\indicator}{\mathbbm{1}}
\newcommand{\TVdist}{d_{\mathsf{TV}}}

\newcommand{\OR}{\mathsf{OR}}
\DeclareMathOperator{\DRE}{DRE}
\newcommand{\Pyes}{\mathcal{P}_{\textsf{yes}}}
\newcommand{\Pno}{\mathcal{P}_{\textsf{no}}}

\newcommand{\poly}{\mathsf{poly}}
\newcommand{\AC}{\mathsf{AC}}
\newcommand{\NC}{\mathsf{NC}}

\newcommand{\Nbb}{\mathbb{N}}
\newcommand{\hsf}{\mathsf{h}}
\newcommand{\Rsf}{\mathsf{R}}
\newcommand{\Psf}{\mathsf{P}}
\newcommand{\Acal}{\mathcal{A}}
\newcommand{\Ecal}{\mathcal{E}}
\newcommand{\Hcal}{\mathcal{H}}
\newcommand{\Jcal}{\mathcal{J}}

\newcommand{\textdef}[1]{\emph{#1}}

\title{Improved lower bounds for decomposable randomized encoding}
\author{Justin Holmgren\thanks{Email: \texttt{justin.holmgren@gmail.com}. Research conducted in part while employed at NTT Research.} \and Kewen Wu\thanks{Caltech. Email: \texttt{shlw\_kevin@hotmail.com}. Research conducted in part while interning at NTT Research.}}
\date{}
\begin{document}
\maketitle

\begin{abstract}
  A decomposable randomized encoding (DRE) for a function $f$ allows $n$ parties, using shared randomness, to encode their individual inputs locally so that the collection of encodings reveals $f(x_1,\ldots,x_n)$ and nothing else.  DREs are widely used in efficient multiparty computation. Their main complexity measure is \emph{size}, the total bit length of the local encodings. Yet the optimal DRE size remains poorly understood even for the $n$-bit OR function.

  We prove the first superlinear lower bound for OR and, more generally, for every non-periodic symmetric function.  Under an additional symmetry assumption, we prove a sharp $\Omega(n\log n)$ lower bound for OR, matching the classic construction of Feige, Kilian, and Naor (STOC 1994).

  We also prove the first $\Omega(n^2)$ lower bound on DRE size for non-explicit Boolean functions.
\end{abstract}

\tableofcontents
\clearpage

\section{Introduction}\label{sec:introduction}

Randomized encodings, introduced by Applebaum, Ishai, and Kushilevitz~\cite{AIK06}, are a beautiful relaxation of deterministic computation.  The notion, of which there are several flavors, is that
\begin{definition}[Informal]
    \label{def:efficient-re}
A randomized function $\hat f$ encodes a deterministic function $f$ if it satisfies the following properties:
\begin{itemize}
    \item \textsc{Correctness.} If $f(x)\ne f(x')$, then $\hat f(x)$ and $\hat f(x')$ have disjoint supports.  Moreover, there is an efficient algorithm $D$, called the \textdef{decoder}, such that $D(\hat f(x))=f(x)$ with probability $1$ (or with high probability).
    \item \textsc{Privacy.} If $f(x)=f(x')$, then $\hat f(x)$ and $\hat f(x')$ are identically distributed.  Moreover, there is an efficient randomized algorithm $S$, called the \textdef{simulator}, such that $S(f(x))$ is identically distributed (or statistically close, or computationally indistinguishable) to $\hat f(x)$.
\end{itemize}
\end{definition}

In this paper, we focus on \emph{decomposable} randomized encodings (DREs, also called garbling schemes), which are randomized encodings of the form
\[
  \hat{f}(x; r) = \big (\hat{f}_1(x_1; r), \ldots, \hat{f}_n(x_n; r) \big ).
  \]
DREs originated in secure two-party computation~\cite{yao1986generate,Kilian88} and have since become indispensable throughout cryptography; see~\cite{Applebaum17,BellareHR12,Ishai13} for surveys.  We consider the regime where $n$ grows and each $x_i$ is binary.  The complementary regime, where the local inputs range over much larger domains, is widely studied in work on private simultaneous messages (PSM)~\cite{FKN94,IK97,BKN18,AssoulineLiu21,AHMS20}.

For several fundamental complexity measures, such as polynomial degree~\cite{AIK06}, the minimum complexity of $\hat f$ can be far smaller than that of $f$. Perhaps the simplest example is the $n$-bit parity function, which cannot be computed in $\AC^0$~\cite{Hastad14} but admits a randomized encoding in $\NC^0$~\cite{Kilian88}. For $n\ge2$, $x\in\{0,1\}^n$, and $r\in\{0,1\}^{n-1}$, define $\widehat{\mathsf{XOR}}(x;r)\in\{0,1\}^n$ by
\[
\widehat{\mathsf{XOR}}(x; r)_i = \begin{cases} x_1 \oplus r_1 & \text{if $i = 1$,} \\
r_{i-1} \oplus x_i \oplus r_i  & \text{if $1 < i \le n-1$,} \\
r_{n-1} \oplus x_n & \text{if $i = n$,}
\end{cases}
\]
and decode by taking the parity of the $n$ encoding bits.  This complexity gap is crucial in recent work on locally sampleable distributions and quantum--classical separations~\cite{KaneOstuniWu25,GrierEtAl26}.

For cryptographic applications, the main complexity measure of a DRE is its \textdef{size}, the total bit length of the local encodings. Unfortunately, even for the $n$-bit OR function, the optimal size is unknown: a simple construction of Feige, Kilian, and Naor gives an $O(n\log n)$ upper bound~\cite{FKN94}, whereas the best prior lower bound is only $2n$~\cite{ShinagawaNuida24}.  For arbitrary $n$-bit Boolean functions, the best prior lower bound is $\Omega(n^2/\log n)$, due to Ball, Holmgren, Ishai, Liu, and Malkin~\cite{BallEtAl20}, while the best general upper bound is exponential~\cite{BKN18}.

\subsection{Our results}\label{sec:intro_result}

We study information-theoretic DREs, so our complexity measure ignores the computational cost of producing or decoding an encoding. We begin with perfect DREs, the cleanest setting, and later consider statistical and structural extensions.

In the information-theoretic setting, we directly work on the jointly distributed local messages $Z_i^b:=\hat f_i(b;r)$ induced by the shared randomness $r$, which leads to the following succinct formulation.

\begin{definition}[Perfect decomposable randomized encoding]\label{def:intro-dre}
Let $f\colon\bin^n\to\bin$.  A perfect DRE of $f$ consists of jointly distributed random variables\footnote{For a positive integer $n$, we use $[n]$ to denote $\{1,\ldots,n\}$.} $(Z_i^b)_{i\in[n],b\in\bin}$, each taking values in a finite alphabet $[K]$.
For $x\in\bin^n$, write $Z^x$ to denote
\[
  (Z_i^{x_i})_{i\in[n]}\in[K]^n.
\]
The following conditions must hold:
\begin{itemize}
\item \textsc{Correctness.}
If $f(x)\ne f(x')$, then the supports of $Z^x$ and $Z^{x'}$ are disjoint.
\item \textsc{Privacy.}
If $f(x)=f(x')$, then $Z^x$ and $Z^{x'}$ are identically distributed.
\end{itemize}
The size of the DRE is\footnote{Throughout, logarithms are base two.} defined to be $n\log K$, and $\DRE(f)$ denotes the minimum size of a perfect DRE of $f$.
\end{definition}
We remark that this definition is weaker than \cref{def:efficient-re} in that it does not require efficient sampling of the distribution of any $Z^x$.  Our lower bounds are thus strengthened by ruling out DREs as per \cref{def:intro-dre}.

For perfect DREs, we classify all symmetric functions whose DRE complexity is $O(n)$, and we prove a genuinely quadratic lower bound for general functions.  We then extend and strengthen these results to the setting of imperfect correctness and privacy.

\subsubsection{Lower bounds for perfect DREs}\label{sec:intro-perfect}

For a symmetric function $f\colon\bin^n\to\bin$, write $f(t)$ for its value on inputs of Hamming weight $t$.  We say $f$ has period $a\in[n+1]$ if $f(t)=f(t+a)$ for all $0\le t\le n-a$.

For symmetric functions, our main structural result is as follows.

\begin{restatable}{theorem}{thmintrosymmetric}\label{thm:intro-symmetric}
For every integer $K\ge1$, there exist integers $N_K,a_K\ge1$ such that if $n\ge N_K$ and a symmetric function $f\colon\bin^n\to\bin$ has a perfect DRE with alphabet $[K]$, then $f$ has period $a_K$.
\end{restatable}

Prior to this work, even for the simplest symmetric function, OR, the best lower bound was only $\DRE(\OR_n)\ge2n$~\cite{ShinagawaNuida24}.  In contrast, \Cref{thm:intro-symmetric} immediately yields the superlinear lower bound $\DRE(\OR_n)=\omega(n)$.\footnote{The proof of \Cref{thm:intro-symmetric} actually relies on the superlinear lower bound for OR; see \Cref{sec:overview}.}

Conversely, the PSM construction of~\cite{BGIKMP14} shows that every symmetric function with period $a$ has a perfect DRE over an alphabet of size at most $a!$.  We therefore obtain the following complete qualitative characterization.

\begin{corollary}\label{cor:intro-symmetric}
A symmetric function has linear DRE if and only if it has constant period.
\end{corollary}

For general Boolean functions, we prove a quadratic worst-case lower bound.

\begin{restatable}{theorem}{thmintrorandom}\label{thm:intro-random}
There exists a universal constant $c>0$ such that, for every $n\ge1$, there exists $f\colon\bin^n\to\bin$ with $\DRE(f)\ge c\cdot n^2$.
\end{restatable}

Prior to this work, the best general lower bound was $\Omega(n^2/\log n)$~\cite{BallEtAl20}; \Cref{thm:intro-random} improves it by a logarithmic factor.

\subsubsection{Statistical, probabilistic, and structural extensions}\label{sec:intro-extensions}

The results in \Cref{sec:intro-perfect} admit several robust extensions.

\paragraph{Statistical extension.}
\Cref{def:intro-dre} defines perfect DREs.  Their statistical counterparts allow small correctness and privacy errors~\cite{AIK06,AIKP15,AHMS20}.

Informally, a statistical DRE specifies two reference distributions, $\Pyes$ and $\Pno$, and requires $Z^x\approx\Pyes$ for $x\in f^{-1}(1)$ and $Z^x\approx\Pno$ for $x\in f^{-1}(0)$.  The correctness error $\delta$ measures the overlap between $\Pyes$ and $\Pno$, while the privacy error $\eta$ is the maximum distance between $Z^x$ and the appropriate reference distribution.  Perfect DREs correspond to $\delta=\eta=0$.  Formal definitions and a comparison with prior notions appear in \Cref{sec:statistical-extensions}.

We strengthen \Cref{thm:intro-symmetric} in this model.

\begin{restatable}{theorem}{thmintroapproxsymmetric}\label{thm:approx-symmetric}
For every integer $K\ge1$, there exist $\varepsilon_K\in(0,1]$ and integers $N_K,a_K\ge1$ such that the following holds.
Assume $n\ge N_K$ and a symmetric function $f\colon\bin^n\to\bin$ has a statistical DRE with alphabet $[K]$, correctness error $\delta$, and privacy error $\eta$.
If $\delta+\eta\le\varepsilon_K$, then $f$ has period $a_K$.
\end{restatable}

Together with \Cref{cor:intro-symmetric}, \Cref{thm:approx-symmetric} yields the following qualitative equivalence.

\begin{corollary}\label{cor:approx_symmetric}
A symmetric function has linear statistical DRE (with sufficiently small error) if and only if it has linear perfect DRE, and these conditions hold if and only if the function has constant period.
\end{corollary}

\paragraph{Probabilistic extension.}
The $\Omega(n^2/\log n)$ lower bound of~\cite{BallEtAl20} holds for a random function. We likewise strengthen \Cref{thm:intro-random} to a random function, even in the statistical setting.

\begin{restatable}{theorem}{thmintroapproxrandom}\label{thm:intro-approx-random}
For every $\bar\delta,\bar\eta\in[0,1]$ with $\bar\delta+2\cdot\bar\eta<1/2$, there is $c_{\bar\delta,\bar\eta}>0$ such that the following holds.
For a uniformly random Boolean function $f\colon\bin^n\to\bin$, except with probability at most
\[
  2^{-c_{\bar\delta,\bar\eta}\cdot 2^n/n^2},
\]
every statistical DRE of $f$ with correctness error $\delta\le\bar\delta$ and privacy error $\eta\le\bar\eta$ has size at least
$$
c_{\bar\delta,\bar\eta}\cdot n^2.
$$
\end{restatable}

\paragraph{A structured lower bound for OR.}
The proofs of \Cref{thm:intro-symmetric} and \Cref{thm:approx-symmetric} use corresponding lower bounds for OR as black boxes; see \Cref{sec:overview}.  It is therefore important to understand the DRE complexity of OR itself.

The classic construction of~\cite{FKN94} gives $\DRE(\OR_n)=O(n\log n)$, whereas the prior lower bound was only $\DRE(\OR_n)\ge2n$~\cite{ShinagawaNuida24}.  \Cref{thm:general-or} improves the latter to $\DRE(\OR_n)=\omega(n)$, although the rate hidden by the $\omega(n)$ notation is extremely weak; see \Cref{sec:discussion}.  To narrow the remaining gap, we prove a matching lower bound for perfect DREs under a symmetry assumption on their support, which captures and goes beyond the classic construction of~\cite{FKN94}.

Recall that $Z^{0^n}\in[K]^n$ denotes the randomized encoding of $\OR_n$ on input $0^n$.  We say that a DRE has symmetric support on input $0^n$ if the support of $Z^{0^n}$ is closed under coordinate permutations.

\begin{restatable}{theorem}{thmintropermutationclosed}\label{thm:intro-permutation-closed}
Every perfect DRE of $\OR_n$ with symmetric support on input $0^n$ has size $\Omega(n\log n)$.
\end{restatable}

\subsection{Discussions and AI disclosure}\label{sec:discussion}

We next discuss the quantitative aspects of our results and highlight several open problems.

\paragraph{On the period.}
The period $a_K$ in \Cref{thm:intro-symmetric} and \Cref{thm:approx-symmetric} may be chosen as the least common multiple of $1,2,\ldots,2K$, and hence $a_K=2^{O(K)}$.  Apart from the threshold $N_K$, this dependence is essentially optimal: for a fixed alphabet $[K]$, there are perfect DREs for $\OR_m$ whenever $1\le m\le\Theta(K)$~\cite{FKN94}.  Consequently, every such $m$ must divide the universal period $a_K$, giving $a_K=2^{\Omega(K)}$.

Our proof actually gives a sharper function-dependent statement: if a symmetric function $f$ has a perfect DRE with alphabet $[K]$, then $f$ has a period $a_f$ satisfying $1\le a_f\le2K$.  The OR example again shows that this bound is essentially optimal.

\paragraph{On the input length.}
The threshold $N_K$ in \Cref{thm:intro-symmetric} is enormous: even for OR, it is not primitive recursive.  The proof reduces the general symmetric case to the OR lower bound with an exponential, but still elementary, loss.  The OR argument itself uses Higman's lemma, whose quantitative bounds are necessarily not primitive recursive.  Consequently, our bound $\DRE(\OR_n)=\omega(n)$ is numerically uninformative.

We nevertheless conjecture the optimal bound below and view the structured lower bound in \Cref{thm:intro-permutation-closed} as supporting evidence.

\begin{conjecture}\label{conj:or-optimal}
$\DRE(\OR_n)=\Omega(n\log n)$.
\end{conjecture}

Regarding \Cref{conj:or-optimal}, we remark that many existing DRE constructions have a group-theoretic flavor~\cite{FKN94,BGIKMP14,Yoshida24,HiwatashiEriguchi26} and are based on permutation branching programs~\cite{Kilian88,Barrington89,IK97}.  Known complexity-theoretic lower bounds therefore apply to this class of constructions. For example, the analysis of Cai and Lipton~\cite{CaiLipton94} confirms \Cref{conj:or-optimal} for constructions based on constant-read permutation branching programs, a class that includes the classic construction~\cite{FKN94}.

\paragraph{On the error dependence.}
\Cref{thm:approx-symmetric} requires both the correctness error $\delta$ and the privacy error $\eta$ to be sufficiently small as functions of $K$.  Ishai~\cite{Ishai13} describes a simple statistical DRE for $\OR_n$ with alphabet $[K]$ for arbitrarily large $n$. With suitable choices of $\Pyes$ and $\Pno$, this construction has either $(\delta,\eta)=(\Theta(1/K),0)$ or $(\delta,\eta)=(0,\Theta(1/K))$. Thus the hypothesis of \Cref{thm:approx-symmetric} must depend on both errors and on $K$.

The value of $\varepsilon_K$ obtained in \Cref{thm:approx-symmetric} is itself not primitive recursive in $K$, again because of the weak quantitative bound for OR.

\paragraph{On the worst-case bound.}
\Cref{thm:intro-random} gives an existential quadratic lower bound for Boolean functions.  This is still exponentially smaller than the upper bound $\DRE(f)=\widetilde O\pbra{2^{n/2}}$~\cite{BKN18}.  We expect the worst-case complexity to be much larger and make the following modest conjecture.

\begin{conjecture}\label{conj:intro-random}
There exists $f\colon\bin^n\to\bin$ with $\DRE(f)=n^{\omega(1)}$.
\end{conjecture}

A further limitation of \Cref{thm:intro-random} is its nonconstructive proof, whereas \cite{BallEtAl20} give explicit functions, including Element Distinctness and Clique, with $\Omega(n^2/\log n)$ lower bounds.  Finding an explicit function of genuinely quadratic DRE complexity remains open.  The same work also constructs a quadratic-size DRE for a candidate pseudorandom function; assuming the candidate is secure, this creates a natural-proofs barrier to constructive superquadratic lower bounds.  Curiously, our non-explicit method also stops at quadratic.

\paragraph{On linear DRE complexity.}
\Cref{cor:intro-symmetric,cor:approx_symmetric} classify only symmetric functions of linear DRE complexity, equivalently those admitting DREs over constant-size alphabets.  It remains to characterize the general Boolean functions that admit small DREs.  We record two relevant observations.

Every nondegenerate monotone Boolean function embeds an OR or AND of superconstant size~\cite{Simon82}.  Our OR lower bound therefore rules out linear-size DREs for such functions.

It is also known that every nondegenerate Boolean function embeds a nontrivial symmetric function of superconstant size~\cite{SunEtAl20}, to which \Cref{thm:intro-symmetric} applies.  However, we do not know how to combine this local information into a global characterization.

\paragraph{AI disclosure.}
The authors proved several structured lower bounds for OR, including \Cref{thm:intro-permutation-closed}, in fall 2022.  More recently, the authors revisited the project with GPT 5.6 Sol and developed the remaining results.  The key new observations \Cref{lem:local-labels,lem:local-anti}, while not hard in hindsight, were discovered by GPT.  GPT 5.6 Sol was also used to produce the initial draft, which the authors subsequently rewrote in full.  The authors remain responsible for all mathematical claims and exposition.

\section{Proof overview}\label{sec:overview}

This section gives a high-level overview of our proof techniques.

\paragraph{A common combinatorial recipe.}

Fix a perfect DRE of $f\colon\bin^n\to\bin$ with alphabet $[K]$. We first show that it yields a compact representation of local values of $f$.

Fix $x\in f^{-1}(1)$. For each $J\subseteq[n]$, let $x\oplus J$ be the string obtained by flipping the bits of $x$ in $J$. Under the same shared randomness, the encodings $Z^x$ and $Z^{x\oplus J}$ differ only in coordinates in $J$. Thus, once $Z^x$ is fixed, correctness (\Cref{def:intro-dre}) allows $f(x\oplus J)$ to be determined from $J$ and the replacement symbols on those coordinates. Privacy ensures that, for every $x'\in f^{-1}(1)$, one can choose randomness under which $Z^{x'}$ equals the same fixed transcript. The same local decoding rule therefore applies throughout the fiber $f^{-1}(1)$.

Applying the same argument to $f^{-1}(0)$ yields local decoders $D_{J,b}\colon[K]^J\to\bin$ and local symbols $L_i\colon\bin^n\to[K]$ such that
\begin{itemize}
\item $L_i(x)$ records the encoding symbol corresponding to $1-x_i$ under the chosen randomness;
\item $D_{J,b} \big ( (L_i(x))_{i\in J} \big ) = f(x\oplus J)$ for every $J\subseteq[n]$ and $x\in f^{-1}(b)$.
\end{itemize}

This local representation underlies several of our proofs: a small-alphabet DRE gives a succinct description of $f(x\oplus J)$ in local neighborhoods of $x$. We formalize the perfect case in \Cref{lem:local-labels} and extend the argument to statistical DREs in \Cref{sec:approx-random}.

\paragraph{Quadratic lower bounds for general Boolean functions.}

The local representation yields the quadratic lower bound for general Boolean functions through a counting argument.

Choose a set $W\subseteq\bin^n$ of pairwise Hamming distance at least $2d+1$, where $d=\Theta(1)$ and $|W|\ge2^n/\poly(n)$. Let $R\subseteq\bin^n$ consist of the points at distance exactly $d$ from $W$. The distance condition ensures that every point in $R$ has a unique representation $x\oplus J$, where $x\in W$ and $|J|=d$.

The construction above computes $f(x\oplus J)$ from $f(x)$, $D_{J,f(x)}$, and $L_i(x)$ for $i\in J$. The decoders $D_{J,b}$ may depend on $f$, $J$, and $b$, but not on $x$. As $f$ varies, there are
$$
2^{|R|}=2^{|W|\cdot\binom nd}
$$
possible truth tables on $R$. On the other hand, the number of local descriptions is at most
$$
\underbrace{2^{|W|}}_{\text{$f(x)$ for $x\in W$}}\cdot\underbrace{K^{n\cdot|W|}}_\text{$L_i(x)$ for $x\in W,i\in[n]$}\cdot\underbrace{2^{K^d\cdot2\cdot\binom nd}}_{\text{$D_{J,b}$'s for $b\in\bin,|J|=d$}}.
$$
For any fixed $d\ge2$, comparing these quantities and using $|W|\ge2^n/\poly(n)$ gives $\DRE(f)=n\log K=\Omega(n^2)$.

We give the formal proof of \Cref{thm:intro-random} in \Cref{sec:random-functions}. A robust version of the argument proves \Cref{thm:intro-approx-random}, the average-case quadratic lower bound for statistical DREs, in \Cref{sec:approx-random}.

\paragraph{Classification for symmetric functions with linear DRE.}

We next consider a symmetric function $f\colon\bin^n\to\bin$, writing $f(t)$ for its value on inputs of Hamming weight $t$. We show that if $f$ has a linear-size DRE---equivalently, one with constant alphabet---then $f$ must be periodic. We prove the perfect version, \Cref{thm:intro-symmetric}, in \Cref{sec:symmetric-classification} and its statistical extension, \Cref{thm:approx-symmetric}, in \Cref{sec:approx-classification}.

Both proofs have three steps.
\begin{itemize}
\item\textsc{Local periodicity.}
For every $t$ sufficiently far from $0$ and $n$, we show that $f$ is periodic between $t$ and $t+\Delta$, where $\Delta$ is an arbitrarily large constant.

Among $2K+1$ nearby weights, choose $K+1$ with the same function value $b$, and apply the local representation to their inputs $1^t0^{n-t}$. At each coordinate of their common zero tail, two of the $K+1$ local labels must coincide. By pigeonholing, some pair of weights $t_1,t_2$ has identical local labels on many coordinates. Flipping any $c$ of those coordinates invokes the same decoder $D_{J,b}$ on the same symbols, so $f(t_1+c)=f(t_2+c)$. This gives a period of length $|t_1-t_2|\le2K$.

\item\textsc{Almost periodicity.}
We then show that $f$ is periodic between $\ell=\Theta(1)$ and $r=n-\Theta(1)$.

This follows by patching together overlapping local periodic intervals. The intervals can begin anywhere away from the boundary and can be made arbitrarily long, allowing their periodic patterns to propagate.

\item\textsc{Full periodicity.}
Finally, we show that the Hamming weights near the boundary obey the periodic pattern established in the interior.

We reduce a boundary mismatch to the OR function. Since $f$ is periodic between $\ell=\Theta(1)$ and $r=n-\Theta(1)$, a mismatch below weight $\ell$ implies
$$
f(t)\ne f(t+a)=f(t+2\cdot a)=\cdots=f(t+m\cdot a),
$$
where $t=O(1)$ is the largest outlier weight below $\ell$ and $a=O(1)$ is the interior period. Standard restrictions and projections reduce this pattern to $\OR_m$ or its negation. For sufficiently large $n$, the resulting $m$ contradicts our superlinear DRE lower bound for OR, which we sketch next.
\end{itemize}

\paragraph{Superlinear lower bounds for OR.}

For OR, perfect privacy makes every nonzero input induce the same encoding distribution $\Pyes$, whose support is disjoint from that of the distribution $\Pno$ induced by $0^n$. Fix a reference string $q=(q_1,\ldots,q_n)$ in the support of $\Pyes$. For each $\emptyset\ne S\subseteq[n]$, we construct a string $p^S\in[K]^n$ satisfying
\begin{equation}\label{eq:overview-local-anti_0}
p^S_i=q_i\quad\text{if and only if}\quad i\notin S.
\end{equation}
To obtain $p^S$, let $e_S\in\bin^n$ be the indicator vector of $S$, choose randomness for which $Z^{e_S}=q$, and let $p^S$ be the corresponding encoding $Z^{0^n}$.

More importantly, these strings satisfy the anti-embedding property
$$
  (p_i^J)_{i\in J}\ne(p_i^L)_{i\in J}\quad\text{for every $\emptyset\ne J\subsetneq L$.}
$$
Otherwise, switching on the coordinates $L\setminus J$ would produce a string in the supports of both $\Pyes$ and $\Pno$, contradicting correctness.

We have thus produced $2^n-1$ strings over $[K]$ with the anti-embedding property. Extremal-combinatorial tools bound $n$ in terms of $K$, implying $\DRE(\OR_n)=\omega(n)$. Concretely, if $n$ were arbitrarily large, K\"onig's lemma (\Cref{lem:konig-infinity}) and the hypergraph Ramsey theorem (\Cref{lem:infinite-ramsey}) would yield an arbitrarily long sequence of strings over $[K]$ in which no string is a subsequence of another, contradicting Higman's lemma (\Cref{lem:higman}).

We formalize this argument for perfect DREs in \Cref{sec:or-lower-bounds}. Its finitary combinatorial form extends to statistical DREs in \Cref{sec:approx-classification}.

\paragraph{Structured optimal lower bound for OR.}
Finally, we outline the optimal $\Omega(n\log n)$ lower bound for OR under the symmetry assumption of \Cref{thm:intro-permutation-closed}; see \Cref{sec:structured-or} for the proof.

Assume that the randomized encoding $Z^{0^n}$ has symmetric support. Together with perfect correctness, this implies that there is no nonempty $S\subseteq[n]$ for which the multisets $\{Z_i^0\}_{i\in S}$ and $\{Z_i^1\}_{i\in S}$ are equal: otherwise, $Z^{0^n}$ would equal $Z^{e_S}$ up to a permutation.

If the alphabet is too small, however, such equal multisets are unavoidable. A greedy argument starts from a singleton $S$ and adds coordinates while maintaining a multiset difference of size $2$, stopping when the multisets agree or no coordinate can be added. Privacy prevents the symbol histograms of $(Z_i^b)_{i\in[n],b\in\bin}$ from being too skewed, so the process must end with equal multisets. The proof in \Cref{sec:structured-or} packages this argument as a topological ordering of an acyclic graph.

\section{Lower bounds for general Boolean functions}\label{sec:random-functions}

In this section, we prove \Cref{thm:intro-random}, restated below.

\thmintrorandom*

We begin by showing how a perfect DRE locally represents the truth table of $f$.

\begin{lemma}\label{lem:local-labels}
Fix an arbitrary $f\colon\bin^n\to\bin$ and an arbitrary perfect DRE of $f$ with alphabet $[K]$.
There exist functions
\begin{itemize}
\item $D_{J,b}\colon[K]^J\to\bin$, where $J\subseteq[n]$ and $b\in\bin$,
\item $L_i\colon\bin^n\to[K]$, where $i\in[n]$,
\end{itemize}
such that
\begin{equation}\label{eq:local-representation}
  f(x\oplus J)
  =D_{J,f(x)}\pbra{(L_i(x))_{i\in J}}
  \quad\text{holds for all $J\subseteq[n]$ and $x\in\bin^n$.}
\end{equation}
Here $x\oplus J$ is $x$ with bits in $J$ flipped.
\end{lemma}
\begin{proof}
The statement is trivial if $f$ is a constant function. Hence we now assume $f^{-1}(0)$ and $f^{-1}(1)$ are both nonempty.

Recall \Cref{def:intro-dre}. All $Z^x$ with $f(x)=1$ (resp., $f(x)=0$) have the same distribution $\Pyes$ (resp., $\Pno$) over $[K]^n$.
We will construct $D_{J,1}$'s and $L_i(x)$ for $x\in f^{-1}(1),i\in[n]$; the complementary case is analogous.

Fix an arbitrary $z$ in the support of $\Pyes$.
For each $x\in f^{-1}(1)$, we fix a choice of $(Z_i^0,Z_i^1)_{i\in[n]}$ with $Z^x=z$ and define
$$
L_i(x)=Z_i^{1-x_i}.
$$
For each $J\subseteq[n]$, define
$$
D_{J,1}(u_J)=\indicator[(z_{[n]\setminus J},u_J)\text{ is in the support of }\Pyes],
$$
where $\indicator[\Ecal]$ is the indicator of the event $\Ecal$.
Then
\begin{align*}
D_{J,1}\pbra{(L_i(x))_{i\in J}}
&=\indicator[(z_{[n]\setminus J},(L_i(x))_{i\in J})\text{ is in the support of }\Pyes]\\
&=\indicator[(z_{[n]\setminus J},(Z_i^{1-x_i})_{i\in J})\text{ is in the support of }\Pyes]
\tag{by the definition of $L_i(x)$}\\
&=\indicator[((Z_i^{x_i})_{i\in [n]\setminus J},(Z_i^{1-x_i})_{i\in J})\text{ is in the support of }\Pyes]
\tag{by the choice of $Z^x=z$}\\
&=\indicator[(Z^{x\oplus J})\text{ is in the support of }\Pyes]
\tag{by the definition of $x\oplus J$}\\
&=f(x\oplus J).
\tag{by the correctness of \Cref{def:intro-dre}}
\end{align*}
This completes the proof.
\end{proof}

\Cref{lem:local-labels} compresses local neighborhoods: instead of recording every value $f(x\oplus J)$, it suffices to record the decoders $D_{J,b}$ and the labels $L_i(x)$. We now use this representation to prove \Cref{thm:intro-random}.

\begin{proof}[Proof of \Cref{thm:intro-random}]
First, if $n<1/c$, then the target lower bound is less than $n$, so any DRE violating the target bound must have a singleton alphabet. Such a DRE can encode only a constant function and the statement naturally holds.

Now assume $n\ge1/c$ and we will set $c>0$ sufficiently small such that $n$ is sufficiently large.
Assume for contradiction that $\DRE(f)\le n^2/4$ for all $f\colon\bin^n\to\bin$.
By padding dummy alphabet elements, this means every Boolean function has a perfect DRE with alphabet size $K=\lfloor2^{n/4}\rfloor$.

Let $W\subseteq\bin^n$ be a large set of strings with pairwise Hamming distance at least $5$. Note that $W$ can be constructed greedily with size
\begin{equation}\label{eq:random-code-size}
|W|=\Omega\pbra{\frac{2^n}{n^4}}.
\end{equation}
Define $R\subseteq\bin^n$ to be the set of strings of Hamming distance exactly $2$ from some string in $W$. By the definition of $W$, we have
\begin{equation}\label{eq:random-radius-two-size}
|R|=|W|\cdot\binom n2.
\end{equation}

We count the possible restrictions to $R$ of the truth tables of all functions $f$.
On the one hand, there are exactly
\begin{equation}\label{eq:random-truth-table-count}
2^{|R|}.
\end{equation}
On the other hand, we use \Cref{lem:local-labels} to provide an upper bound: for each $f$, we simply record (1) $D_{J,b}$ for $|J|=2$ and $b\in\bin$, (2) $L_i(x)$ for $i\in[n]$ and $x\in W$, and (3) $f(x)$ for $x\in W$.
This upper bounds the number of possibilities by
$$
\pbra{2^{K^2}}^{\binom n2\cdot2}\cdot\pbra{K^{|W|}}^n\cdot2^{|W|}.
$$
Combining with \Cref{eq:random-code-size,eq:random-radius-two-size,eq:random-truth-table-count} and $K=\lfloor2^{n/4}\rfloor$, we have
$$
\pbra{2^{2^{n/2}}}^{\binom n2\cdot2}\cdot\pbra{2^{|W|\cdot n/4}}^n\cdot2^{|W|}\ge2^{|W|\cdot\binom n2}
\quad\text{and}\quad
|W|=\Omega\pbra{\frac{2^n}{n^4}},
$$
which is impossible for $n$ sufficiently large.
\end{proof}

\section{Classification of symmetric functions}\label{sec:symmetric-classification}

Recall that for a symmetric function $f$, we write $f(t)$ for its value on inputs of Hamming weight $t$. We say that $f$ has period $a$ if $f(t)=f(t+a)$ for every $0\le t\le n-a$. This section proves \Cref{thm:intro-symmetric}, restated below.

\thmintrosymmetric*

The proof has three steps. First, we show that $f$ is periodic away from Hamming weights near $0$ and $n$. This is the content of \Cref{lem:periodic-core}, whose proof relies on the local representation in \Cref{lem:local-labels}.

\begin{restatable}{lemma}{lemperiodiccore}\label{lem:periodic-core}
For every integer $K\ge1$, there exist integers $B_K,a_K\ge1$ such that if a symmetric function $f\colon\bin^n\to\bin$ has a perfect DRE with alphabet $[K]$, then $f$ has period $a_K$ between Hamming weights $B_K$ and $n-B_K$. That is,
$$
f(t)=f(t+a_K)\quad\text{holds for all $B_K\le t\le n-B_K-a_K$.}
$$
\end{restatable}

Second, we reduce any failure of this pattern at the boundary to the OR function. We use the following closure properties of perfect DREs.

\begin{restatable}{lemma}{lemblockid}\label{lem:block-identification}
Assume $f\colon\bin^n\to\bin$ has a perfect DRE with alphabet $[K]$.
\begin{itemize}
\item \textsc{Rearranging.}
For any permutation $\pi\colon[n]\to[n]$, $f_\pi\colon\bin^n\to\bin$ has a perfect DRE with alphabet $[K]$, where
$$
f_\pi(x_1,\ldots,x_n):=f(x_{\pi(1)},\ldots,x_{\pi(n)}).
$$
\item \textsc{Reverse.}
$1-f$ has a perfect DRE with alphabet $[K]$.
\item \textsc{Negation.}
$g\colon\bin^n\to\bin$ has a perfect DRE with alphabet $[K]$, where
$$
g(x_1,x_2,\ldots,x_n):=f(1-x_1,x_2,\ldots,x_n).
$$
\item \textsc{Restriction.}
$h\colon\bin^{n-1}\to\bin$ has a perfect DRE with alphabet $[K]$, where
$$
h(x_1,\ldots,x_{n-1}):=f(x_1,\ldots,x_{n-1},1).
$$
\item \textsc{Projection.}
Partition $[n]$ arbitrarily into $m$ sets $E_1,E_2,\ldots,E_m$, each of size at most $d$.

Then $F\colon\bin^m\to\bin$ has a perfect DRE with alphabet $[K^d]$, where
$$
F(x_1,\ldots,x_m):=f(y_1,\ldots,y_n)
\quad\text{and}\quad
\text{$y_i=x_j$ if $i\in E_j$.}
$$
\end{itemize}
\end{restatable}

The final ingredient is the following lower bound for OR.

\begin{restatable}{theorem}{thmgeneralor}\label{thm:general-or}
For every integer $K\ge1$, there exists an integer $n_K\ge1$ such that if $\OR_n$ has a perfect DRE with alphabet $[K]$, then $n\le n_K$.
Consequently, $\DRE(\OR_n)=\omega(n)$.
\end{restatable}

We prove \Cref{lem:periodic-core}, \Cref{lem:block-identification}, and \Cref{thm:general-or} in \Cref{sec:fixed-alphabet-periodicity,sec:reduction,sec:or-lower-bounds}, respectively.
We now combine these ingredients to prove the classification theorem.

\begin{proof}[Proof of \Cref{thm:intro-symmetric}]
Define
\begin{equation}\label{eq:thm:intro-symmetric_1}
N_K=a_K\cdot\pbra{1+n_{K^{a_K}}}+2\cdot B_K+a_K.
\end{equation}
Here $n_K$ comes from \Cref{thm:general-or}, while $a_K$ and $B_K$ come from \Cref{lem:periodic-core}.

Assume for contradiction that $f$ does not have period $a_K$. By \Cref{lem:periodic-core}, the mismatch can only come from boundary Hamming weights, leading to the following cases.

\paragraph{Mismatch from small weights.}
The periodic interior defines a value for every residue modulo $a_K$. In a residue class with an exception below $B_K$, choose $w$ to be the largest such exception, and put $m=\lfloor(n-B_K-w)/a_K\rfloor$. Then $0\le w<B_K$ and
$$
f(w)\ne f(w+a_K)=f(w+2\cdot a_K)=\cdots=f(w+m\cdot a_K),
$$
where $m\ge n_{K^{a_K}}+1$ by \Cref{eq:thm:intro-symmetric_1}.
We will show how to obtain a perfect DRE for $\OR_m$.
By \Cref{lem:block-identification}, it suffices to construct $\OR_m$ from $f$ as follows.
\begin{itemize}
\item Using \textsc{Restriction} and \textsc{Negation} rules on $f$, we fix $w$ ones and $n-w-a_K\cdot m$ zeros, which gives a symmetric $f'\colon\bin^{m\cdot a_K}\to\bin$ satisfying $f'(0)\ne f'(a_K)=f'(2\cdot a_K)=\cdots=f'(m\cdot a_K)$.
\item Using \textsc{Projection} rule on $f'$, we partition the input bits into blocks of size $a_K$, which gives a symmetric $f''\colon\bin^m\to\bin$ satisfying $f''(0)\ne f''(1)=f''(2)=\cdots=f''(m)$.
\item At this point, $f''$ is either $\OR_m$ or $1-\OR_m$, and we simply apply \textsc{Reverse} rule if necessary.
\end{itemize}
By \Cref{lem:block-identification}, we obtain a perfect DRE for $\OR_m$ with alphabet $[K^{a_K}]$, which contradicts \Cref{thm:general-or}.

\paragraph{Mismatch from large weights.}
Complementing every input bit replaces the weight sequence $f(t)$ by $f(n-t)$ and preserves the common alphabet by \Cref{lem:block-identification}. A boundary exception above $n-B_K$ becomes one below $B_K$, reducing to the preceding case.
\end{proof}

\subsection{Almost periodicity}\label{sec:fixed-alphabet-periodicity}

To prove \Cref{lem:periodic-core}, we directly define
\begin{equation}\label{eq:sec:fixed-alphabet-periodicity_1}
a_K=\lcm(1,2,\ldots,2K).
\end{equation}

We begin by establishing local periodicity.

\begin{lemma}\label{lem:local-periodicity}
Let $f\colon\bin^n\to\bin$ be symmetric and have a perfect DRE over $[K]$.
For every $0\le u\le n-2K$, $f$ has period $a_K$ between Hamming weights $\ell$ and $r$, where
\[
  \ell=u+2K
  \quad\text{and}\quad
  r=u+1+\frac{n-u-2K}{\binom{K+1}2}.
\]
Consequently, if $0\le u\le n-2K-\binom{K+1}2\cdot\pbra{2K+a_K}$, then $r\ge\ell+a_K+1$.
\end{lemma}
\begin{proof}
The statement is vacuously true if $n\le2K$; hence we assume $n\ge2K+1$.
Among the $2K+1$ weights between $u$ and $u+2K$, at least $K+1$ have the same function value in $f$.
Choose
\begin{equation}\label{eq:lem:local-periodicity_0}
  u\le s_0<s_1<\cdots<s_K\le u+2K
\end{equation}
and $b\in\bin$ such that $f(s_j)=b$ for every $j$.
For each $t\in\{s_0,\ldots,s_K\}$, let $x^{(t)}=1^t0^{n-t}$. Applying \Cref{lem:local-labels}, we obtain functions
\begin{itemize}
    \item $D_{J,b}\colon[K]^J\to\bin$, where $J\subseteq[n]$,
    \item $L_i\colon\bin^n\to[K]$, where $i\in[n]$,
\end{itemize}
such that
\begin{equation}\label{eq:lem:local-periodicity_1}
f(x^{(t)}\oplus J)=D_{J,b}\pbra{(L_i(x^{(t)}))_{i\in J}}
\quad\text{holds for all $J\subseteq[n]$ and $t\in\{s_0,\ldots,s_K\}$.}
\end{equation}
Now for each $u+2K+1\le i\le n$, the $K+1$ symbols $L_i(x^{(s_0)}),\ldots,L_i(x^{(s_K)})$ must collide.
Hence there exist $H\subseteq\{u+2K+1,\ldots,n-1,n\}$ of size
\begin{equation}\label{eq:lem:local-periodicity_2}
|H|\ge\frac{n-u-2K}{\binom{K+1}2}
\end{equation}
and some distinct $s,s'\in\{s_0,\ldots,s_K\}$ such that
\begin{equation}\label{eq:lem:local-periodicity_3}
L_i(x^{(s)})=L_i(x^{(s')})
\quad\text{for all $i\in H$.}
\end{equation}
Since $H\subseteq\{u+2K+1,\ldots,n-1,n\}$, we know that both $x^{(s')}$ and $x^{(s)}$ are all-zero on coordinates in $H$.
Hence for every $0\le p\le|H|$, we can choose $J\subseteq H$ of size $|J|=p$ and obtain
\begin{align*}
f(s+p)
&=f(x^{(s)}\oplus J)
=D_{J,b}\pbra{(L_i(x^{(s)}))_{i\in J}}
\tag{by \Cref{eq:lem:local-periodicity_1}}\\
&=D_{J,b}\pbra{(L_i(x^{(s')}))_{i\in J}}
\tag{by \Cref{eq:lem:local-periodicity_3}}\\
&=f(x^{(s')}\oplus J)=f(s'+p).
\tag{by \Cref{eq:lem:local-periodicity_1}}
\end{align*}
Without loss of generality, assume $s'<s$.
This means $f$ has period $s-s'$ between weights $s'$ and $s+|H|$.
Since $1\le s-s'\le 2K$, it divides $a_K$ by \Cref{eq:sec:fixed-alphabet-periodicity_1}. We may therefore take the period to be $a_K$.
In addition, \Cref{eq:lem:local-periodicity_0,eq:lem:local-periodicity_2} and our choice of $\ell,r$ give $s'\le\ell$ and $r\le s+|H|$, which completes the proof.
\end{proof}

We now combine the local intervals to prove \Cref{lem:periodic-core}.

\begin{proof}[Proof of \Cref{lem:periodic-core}]
Fix $K$ and set
$$
B_K=2K+\binom{K+1}2\cdot\pbra{2K+a_K}.
$$
The statement is vacuously true if $n\le B_K$; hence we assume $n\ge B_K$.

For each $0\le u\le n-B_K$, we apply \Cref{lem:local-periodicity} and obtain that $f$ has period $a_K$ between weights $\ell_u$ and $r_u$ where
\begin{equation}\label{eq:lem:periodic-core_1}
  \ell_u=u+2K
  \quad\text{and}\quad
  r_u\ge\ell_u+a_K+1
\end{equation}
Since $[\ell_u,r_u]\cap [\ell_{u+1},r_{u+1}]$ contains the integer interval $[u+2K+1,u+2K+a_K+1]$, consisting of $a_K+1$ weights, the period pattern extends.
As a result, $f$ has period $a_K$ between
$$
\ell_0=2K\le B_K
\quad\text{and}\quad
r_{n-B_K}\ge n-B_K+2K+a_K+1\ge n-B_K
$$
as desired.
\end{proof}

\subsection{Reduction between perfect DREs}\label{sec:reduction}

This subsection proves \Cref{lem:block-identification}, restated below.

\lemblockid*

\begin{proof}
Fix a perfect DRE of $f$ as in \Cref{lem:block-identification}.
Let $\Lambda$ be the probability space, and assume that each $\lambda\in\Lambda$ has nonzero probability mass, so each $\lambda\in\Lambda$ determines $Z(\lambda)=(Z_i^b(\lambda))_{i\in[n],b\in\bin}$ in \Cref{def:intro-dre}.
As such, the DRE of $f$ given input $x$ is computed by first drawing $\lambda\in\Lambda$ and then outputting $Z^x=Z^x(\lambda)$ deterministically.

Now we handle each reduction separately and provide the corresponding DRE $\widetilde Z$.

\paragraph{Rearranging.}
We rearrange $\widetilde Z^b_j=Z^b_{\pi^{-1}(j)}$ as a perfect DRE of $f_\pi$.

\paragraph{Reverse.}
We directly set $\widetilde Z=Z$ as a perfect DRE of $1-f$.

\paragraph{Negation.}
Set $\widetilde Z_1^b=Z_1^{1-b}$ and $\widetilde Z_i^b=Z_i^b$ for $i\ge2$.
Then $\widetilde Z$ is a perfect DRE of $g$.

\paragraph{Restriction.}
Fix an arbitrary $\lambda_\star\in\Lambda$ and define $q_\star=Z_n^1(\lambda_\star)$.
Then for each $\lambda\in\Lambda$ with $Z_n^1(\lambda)=q_\star$, define $\widetilde Z^b_i(\lambda)=Z_i^b(\lambda)$ for $i\in[n-1]$.
Then $\widetilde Z$ is a perfect DRE of $h$: given input $(x_1,\ldots,x_{n-1})$, the DRE is computed by first drawing $\lambda\in\Lambda$ conditioned on $Z_n^1(\lambda)=q_\star$, and then outputting $\widetilde Z(\lambda)$.

To see that $\widetilde Z$ is a perfect DRE, we verify conditions in \Cref{def:intro-dre} directly.
\begin{itemize}
\item If $h(x_1,\ldots,x_{n-1})=h(x'_1,\ldots,x'_{n-1})$, then $Z^{(x_1,\ldots,x_{n-1},1)}$ and $Z^{(x'_1,\ldots,x'_{n-1},1)}$ are identically distributed. Conditioning both on their last coordinate being $q_\star$ preserves this equality in distribution. After deleting that coordinate, the resulting distributions are precisely $\widetilde Z^{(x_1,\ldots,x_{n-1})}$ and $\widetilde Z^{(x'_1,\ldots,x'_{n-1})}$.
\item If the two values of $h$ differ, correctness implies that $Z^{(x_1,\ldots,x_{n-1},1)}$ and $Z^{(x'_1,\ldots,x'_{n-1},1)}$ have disjoint supports. Because their last coordinates both equal $q_\star$ after conditioning, their first $n-1$ coordinates must differ. Thus the corresponding values of $\widetilde Z$ also have disjoint supports.
\end{itemize}

\paragraph{Projection.}
For each $E_j$ and $b\in\bin$, group the corresponding local messages into
\[
  \widetilde Z_j^b=(Z_i^b)_{i\in E_j}.
\]
This is a string in $[K]^{|E_j|}$ and we naturally embed it as a number in $[K^d]$ since $|E_j|\le d$.

To see that $\widetilde Z$ is a perfect DRE of $F$, for each $x\in\bin^m$ we observe that $\widetilde Z^x$ is simply another way of writing $Z^y$, where $y\in\bin^n$ is given by $y_i=x_j$ when $i\in E_j$. Hence in light of \Cref{def:intro-dre}, all desired properties of $\widetilde Z$ follow directly from those of $Z$.
\end{proof}

\subsection{Lower bounds for OR}\label{sec:or-lower-bounds}

We now prove the remaining ingredient for \Cref{thm:intro-symmetric}: a superlinear DRE lower bound for $\OR_n$.

\thmgeneralor*

\Cref{thm:general-or} follows from two lemmas. The first extracts a local anti-embedding family from a perfect DRE, and the second gives a purely combinatorial bound on the dimension of such a family.

\begin{lemma}\label{lem:anti-embedding_family}
Assume $\OR_n$ has a perfect DRE over alphabet $[K]$.
Then there exists a family of strings $p^S\in[K]^n$, indexed by $S\subseteq[n]$, such that the following holds.
\begin{itemize}
\item $p^J_J\ne p_J^L$ for every $\emptyset\ne J\subsetneq L\subseteq[n]$.
\item $p^S_i=p^\emptyset_i$ for every $S\subseteq[n]$ and $i\notin S$.
\end{itemize}
\end{lemma}

\begin{lemma}\label{lem:nonexist_anti}
For every integer $K\ge1$, there exists an integer $n_K\ge1$ such that if the family of strings in \Cref{lem:anti-embedding_family} exists, then $n\le n_K$.
\end{lemma}

We prove these lemmas in \Cref{sec:anti-emb,sec:anti-obstruction}, respectively.

\begin{proof}[Proof of \Cref{thm:general-or}]
Simply combine \Cref{lem:anti-embedding_family} and \Cref{lem:nonexist_anti}.
\end{proof}

\subsubsection{Local anti-embedding family}\label{sec:anti-emb}

To prove \Cref{lem:anti-embedding_family}, we fix an arbitrary perfect DRE of $\OR_n$ with alphabet $[K]$.
For simplicity, we use $\Lambda$ to denote the probability space of the DRE and assume that each $\lambda\in\Lambda$ has nonzero probability mass, so each $\lambda$ determines $Z(\lambda)=(Z_i^b(\lambda))_{i\in[n],b\in\bin}$ in \Cref{def:intro-dre}.
As such, the DRE of $\OR_n$ given input $x$ is computed by first drawing $\lambda\in\Lambda$ and then outputting $Z^x=Z^x(\lambda)$ deterministically.

We start with the following simple observation.

\begin{lemma}\label{lem:pointwise-separation}
For every $i\in[n]$ and $\lambda\in\Lambda$, we have $Z_i^0(\lambda)\ne Z_i^1(\lambda)$.
Consequently, $K\ge2$.
\end{lemma}
\begin{proof}
If $Z_i^0(\lambda)=Z_i^1(\lambda)$, then $Z^{0^n}(\lambda)=Z^{e_i}(\lambda)$ where $e_i$ is the indicator vector of $i$.
Since $\OR_n(0^n)=0$ and $\OR_n(e_i)=1$, this violates the correctness property of \Cref{def:intro-dre}.
\end{proof}

We will upgrade the above pointwise separation to a stronger local anti-embedding property. This crucially requires us to move between different $\lambda$'s in $\Lambda$.

By \Cref{def:intro-dre}, all $Z^x$ with $x\in\bin^n\setminus\{0^n\}$ have the same distribution $\Pyes$, whose support is disjoint from that of $\Pno$, the distribution of $Z^{0^n}$.
Fix an arbitrary transcript
$$
(q_1,\ldots,q_n)\quad\text{in the support of $\Pyes$.}
$$
For each $\emptyset\ne S\subseteq[n]$, we use $e_S\in\bin^n$ to denote the indicator vector of $S$ and choose an arbitrary $\lambda_S\in\Lambda$ such that $Z^{e_S}(\lambda_S)=q$.
Define
\begin{equation}\label{eq:above_local-anti}
  p^S=Z^{0^n}(\lambda_S),
\end{equation}
which is in the support of $\Pno$.

\begin{lemma}[Local anti-embedding]\label{lem:local-anti}
For every $\emptyset\ne S\subseteq[n]$ and $i\in[n]$, we have
\begin{equation}\label{eq:lem:local-anti_1}
p^S_i=q_i\quad\text{if and only if $i\notin S$.}
\end{equation}
Moreover, if $\emptyset\ne J\subsetneq L\subseteq[n]$, then $p_J^J\ne p_J^L$.
\end{lemma}
\begin{proof}
If $i\notin S$, then the $i$th coordinate of $e_S$ is $0$, which, by \Cref{eq:above_local-anti}, means
$$
p^S_i=\pbra{Z^{0^n}(\lambda_S)}_i=Z_i^0(\lambda_S)=\pbra{Z^{e_S}(\lambda_S)}_i=q_i.
$$
If $i\in S$, then similarly we have $Z_i^1(\lambda_S)=q_i$ and, by \Cref{lem:pointwise-separation}, $p_i^S=Z_i^0(\lambda_S)\ne q_i$.
This proves \Cref{eq:lem:local-anti_1}.

For the ``moreover'' part, assume for contradiction that $p_J^J=p_J^L$.
Define $B=L\setminus J$, which is a nonempty set.
We will show that $Z^{e_B}(\lambda_L)=p^J$; this violates the correctness of \Cref{def:intro-dre} since $\OR_n(e_B)=1$ yet $p^J$ is in the support of $\Pno$.
To see this, we divide $i\in[n]$ into cases.
\begin{itemize}
\item If $i\notin L$, then $\pbra{Z^{e_B}(\lambda_L)}_i=Z^0_i(\lambda_L)=\pbra{Z^{e_L}(\lambda_L)}_i=q_i$. By \Cref{eq:lem:local-anti_1}, we have $p^J_i=q_i$ and hence $\pbra{Z^{e_B}(\lambda_L)}_i=p^J_i$.
\item If $i\in B$, then $\pbra{Z^{e_B}(\lambda_L)}_i=Z^1_i(\lambda_L)=\pbra{Z^{e_L}(\lambda_L)}_i=q_i$. By \Cref{eq:lem:local-anti_1}, we have $p^J_i=q_i$ and hence $\pbra{Z^{e_B}(\lambda_L)}_i=p^J_i$.
\item If $i\in J$, then $\pbra{Z^{e_B}(\lambda_L)}_i=Z^0_i(\lambda_L)=p^L_i$. Since we assumed $p^L_J=p^J_J$, we also have $\pbra{Z^{e_B}(\lambda_L)}_i=p^L_i=p^J_i$.
\end{itemize}
This completes the proof.
\end{proof}

To summarize, any perfect DRE of $\OR_n$ can be turned into a family of strings with anti-embedding structure. This readily proves \Cref{lem:anti-embedding_family}.

\begin{proof}[Proof of \Cref{lem:anti-embedding_family}]
Construct $p^S$ with \Cref{eq:above_local-anti} and define $p^\emptyset=q$.
Then the desired properties follow from \Cref{lem:local-anti}.
\end{proof}

\subsubsection{Finiteness of local anti-embedding family}\label{sec:anti-obstruction}

The proof of \Cref{lem:nonexist_anti} uses three standard infinitary results.

\begin{lemma}[K\H{o}nig's infinity lemma~\cite{Konig27}]\label{lem:konig-infinity}
Every infinite rooted tree in which each vertex has finitely many children contains an infinite path starting at the root.
\end{lemma}

\begin{lemma}[Infinite hypergraph Ramsey theorem~\cite{Ramsey30}]\label{lem:infinite-ramsey}
Let\footnote{We use $\Nbb_+=\cbra{1,2,\ldots}$ to denote the set of positive integers.} $H\subseteq\Nbb_+$ be an arbitrary infinite set.
For all integers $s,r\ge1$, every coloring of the size-$s$ subsets of $H$ with $r$ colors has an infinite set $H'\subseteq H$ all of whose size-$s$ subsets have the same color.
\end{lemma}

\begin{lemma}[Higman's lemma~\cite{Higman52}]\label{lem:higman}
Let $\Sigma$ be a finite alphabet. For every infinite sequence
$w_1,w_2,\ldots$ from $\Sigma^*$, there are indices $a<b$ such that $w_a$ is a subsequence\footnote{That is, $w_a$ can be obtained from $w_b$ by deleting symbols.} of $w_b$.
\end{lemma}

We now prove \Cref{lem:nonexist_anti}. Quantitative versions of the preceding results yield an explicit bound \cite{ErdosRado52,SchmitzSchnoebelen11}, but it is not primitive recursive.

\begin{proof}[Proof of \Cref{lem:nonexist_anti}]
Observe that any valid family of strings $\cbra{p^S}_{S\subseteq[n]}$ in \Cref{lem:anti-embedding_family} of dimension $n$ naturally gives valid families of smaller dimensions $n'<n$, simply by discarding $p^S$ if $S\not\subseteq[n']$ and truncating $p^S$ to the first $n'$ coordinates if $S\subseteq[n']$.
For convenience, we write $\cbra{r^S}_{S\subseteq[n']}\prec\cbra{p^S}_{S\subseteq[n]}$ if the former family is derived from the latter by the method described.

Assume for contradiction that \Cref{lem:nonexist_anti} is false.
By the discussion above, valid families exist in every dimension $n\ge1$.
Given this, we construct a finitely branching rooted forest as follows: for each valid family of dimension $n$, we create a tree vertex at depth $n$; then we connect two vertices if their dimensions differ by exactly $1$ and they satisfy the $\prec$ relation above.
That is, consider two vertices labeled by $\Rsf=\cbra{r^S}_{S\subseteq[n']}$ and $\Psf=\cbra{p^S}_{S\subseteq[n]}$ respectively; then $\Rsf$ is the parent vertex of $\Psf$ if and only if $n'=n-1$ and $\Rsf\prec\Psf$.
The forest has finitely many roots, and, because $[K]$ is finite, each vertex has finitely many children.
Hence we can apply \Cref{lem:konig-infinity} and find an infinite path $\Psf_1\prec\Psf_2\prec\Psf_3\prec\cdots$ starting at some depth-$1$ vertex $\Psf_1$, where each $\Psf_d$ is a valid family of dimension $d$.

We next use this path to define a coloring $\chi$ on all nonempty finite subsets of $\Nbb_+$.
For each such $S$, choose any integer $d$ sufficiently large such that $d\ge\max_{i\in S}i$ and let $p^{S,d}\in\Psf_d$ be the string from $\Psf_d$ indexed by $S$.
Define
\[
  \chi(S)=p^{S,d}_S\in[K]^{|S|},
\]
which lists the symbols on coordinates $S$ in increasing coordinate order.
By the $\prec$ relation defined above, $\chi(S)$ does not depend on the choice of $d$.
Observe that for fixed $s$, the coloring $\chi(S)$ on $|S|=s$ has finitely many colors.
Hence we can sequentially apply \Cref{lem:infinite-ramsey} and homogenize as follows: we start with $s=1$. By \Cref{lem:infinite-ramsey}, we find an infinite $H_1\subseteq\Nbb_+$ and a color $c_1\in[K]$ such that $\chi(S)=c_1$ for every size-$1$ set $S\subseteq H_1$. Then sequentially for $s\ge2$, we apply \Cref{lem:infinite-ramsey} to find an infinite $H_s\subseteq H_{s-1}$ and a color $c_s\in[K]^s$ such that $\chi(S)=c_s$ for every size-$s$ set $S\subseteq H_s$.

At this point, we can obtain an infinite sequence $c_1,c_2,\ldots$ from $[K]^*$.
By \Cref{lem:higman}, there exist indices $d'<d$ such that $c_{d'}\in[K]^{d'}$ is a subsequence of $c_d\in[K]^d$.
Choose indices $1\le j_1<\cdots<j_{d'}\le d$ witnessing this subsequence embedding. Take any $S=\{i_1<\cdots<i_d\}\subset H_d$, set $S'=\{i_{j_1},\ldots,i_{j_{d'}}\}$, and choose $D\ge\max S$. Since $H_d\subseteq H_{d'}$, we have
$$
p^{S',D}_{S'}=c_{d'}\quad\text{is a subsequence of}\quad p_S^{S,D}=c_d.
$$
By the choice of the positions $j_1,\ldots,j_{d'}$, this gives $p^{S',D}_{S'}=p^{S,D}_{S'}$.
Since $p^{S',D}$ and $p^{S,D}$ are strings from the valid family $\Psf_D$, this contradicts the first property in \Cref{lem:anti-embedding_family}, which completes the proof.
\end{proof}

\section{Statistical and probabilistic extensions}\label{sec:statistical-extensions}

This section proves the statistical and probabilistic extensions stated in \Cref{sec:intro-extensions}.

We begin with statistical DREs. For distributions $\mu,\nu$, let $\TVdist(\mu,\nu)$ denote their total variation distance.

\begin{definition}[Statistical decomposable randomized encoding]\label{def:statistical-dre}
Let $f\colon\bin^n\to\bin$.
A statistical DRE of $f$ consists of possibly correlated random variables $(Z_i^b)_{i\in[n],b\in\bin}$, each taking values in a finite alphabet $[K]$, together with two distributions $\Pyes,\Pno$ over $[K]^n$. On input $x$, the encoding is the random vector $Z^x=(Z_i^{x_i})_{i\in[n]}\in[K]^n$.

The correctness error $\delta$ is defined as
$$
\delta=1-\TVdist(\Pyes,\Pno).
$$
The privacy error $\eta$ is defined as
$$
\eta=\max\cbra{\max_{x\in f^{-1}(1)}\TVdist(\mu_x,\Pyes),\max_{x\in f^{-1}(0)}\TVdist(\mu_x,\Pno)},
$$
where $\mu_x$ is the distribution of $Z^x$.
The size of the DRE is $n\log K$ as in \Cref{def:intro-dre}.
\end{definition}

A statistical DRE is perfect exactly when $\delta=\eta=0$.

\begin{remark}
The definitions in \cite{AIK06,AIKP15} use a decoder and a simulator. Ignoring computational efficiency, let $\delta_{\mathrm{dec}}$ be the minimum worst-case decoding error and $\eta_{\mathrm{sim}}$ the minimum simulator-based privacy error. For nonconstant $f$ and any fixed reference distributions in \Cref{def:statistical-dre},
\[
\max\{0,\delta/2-\eta\}\le\delta_{\mathrm{dec}}\le\delta+\eta.
\]
For the upper bound, the likelihood-ratio test on $\Pyes,\Pno$ has the sum of its two reference errors equal to $\delta$; replacing a reference distribution by $\mu_x$ increases the corresponding error by at most $\eta$. Conversely, a decoder with worst-case error $e$ on the actual encodings has error at most $e+\eta$ on each reference distribution, so $\delta\le2(e+\eta)$. For constant $f$, $\delta_{\mathrm{dec}}=0$.

The reference distributions themselves define a simulator with error $\eta$, so $\eta_{\mathrm{sim}}\le\eta$; minimizing $\eta$ over the references gives exactly $\eta_{\mathrm{sim}}$. This optimization can change $\delta$. We use \Cref{def:statistical-dre} because the reference distributions $\Pyes,\Pno$ simplify the proofs.
\end{remark}

\subsection{Average-case lower bounds for general Boolean functions}\label{sec:approx-random}

We now prove \Cref{thm:intro-approx-random}, a robust, high-probability version of \Cref{thm:intro-random}.

\thmintroapproxrandom*

\begin{proof}
The proof largely follows the structure of the proof of \Cref{thm:intro-random}. First, assume $n\ge1/c_{\bar\delta,\bar\eta}$, where we set $c_{\bar\delta,\bar\eta}\in(0,1/2]$ later. If $n<1/c_{\bar\delta,\bar\eta}$, then the target lower bound is less than $n$, so any DRE violating the target bound must have a singleton alphabet. Such a DRE can encode only a constant function, and the two constant functions constitute a $2\cdot2^{-2^n}\le2^{-2^n/2}\le2^{-c_{\bar\delta,\bar\eta}\cdot2^n}\le2^{-c_{\bar\delta,\bar\eta}\cdot2^n/n^2}$ fraction of all Boolean functions, as claimed.

Define $\theta=\bar\delta+2\bar\eta<1/2$ and its binary entropy $\Hcal(\theta)=\theta\log\pbra{\frac1\theta}+(1-\theta)\log\pbra{\frac1{1-\theta}}$, with $\Hcal(0)=0$ by continuity.
Choose $0<c<1/2$ such that $\Hcal(\theta)+2c<1$.
Fix sets $W,R\subseteq\bin^n$ as in the proof of \Cref{thm:intro-random} such that every point of $R$ has a unique representation $x\oplus J$ with $x\in W$
and $|J|=2$.
In addition,
\begin{equation}\label{eq:thm:intro-approx-random_1}
  |W|=\Omega\pbra{\frac{2^n}{n^4}}
  \quad\text{and}\quad
  |R|=|W|\binom n2.
\end{equation}

Fix a function $f$ having a statistical DRE with the stated error bounds and size at most $c\cdot n^2$.
We again use $\Lambda$ for its probability space; the encoding $Z(\lambda)=(Z_i^b(\lambda))_{i\in[n],b\in\bin}$ is deterministic once $\lambda\in\Lambda$ is fixed.
By padding, we assume its alphabet is $[K]$, where $K=\lceil2^{cn}\rceil$.
Let
\[
  \mathcal A=\{z:\Pyes(z)>\Pno(z)\}\subseteq[K]^n,
\]
which witnesses the total variation distance between $\Pyes$ and $\Pno$.

Sample $Y_1\sim\Pyes$ and $Y_0\sim\Pno$.
For every $x\in W$, the definition of $\eta$ allows us to couple $\lambda_x\sim\Lambda$ with $Y_{f(x)}$ such that $\Pr[Z^x(\lambda_x)\ne Y_{f(x)}]\le\bar\eta$.
All these couplings can be realized simultaneously: after sampling $Y_0,Y_1$, sample each $\lambda_x$ from its coupling's conditional distribution given $Y_{f(x)}$. Each $\lambda_x$ retains the original randomness distribution as its marginal.
Define
\[
  L_i(x)=Z_i^{1-x_i}(\lambda_x).
\]
For each $J\subseteq[n]$ with $|J|=2$ and $b\in\bin$, define $D_{J,b}\colon[K]^J\to\bin$ by letting
$$
D_{J,b}(u_J)=\indicator\sbra{((Y_b)_{[n]\setminus J},u_J)\text{ is in }\Acal}.
$$
In \Cref{thm:intro-random}, we always have $D_{J,f(x)}\pbra{(L_i(x))_{i\in J}}=f(x\oplus J)$.
Here, because of the statistical error, the prediction $D_{J,f(x)}\pbra{(L_i(x))_{i\in J}}$ for $f(x\oplus J)$ can be wrong only
\begin{enumerate}
\item\label{itm:thm:intro-approx-random_1}
if $Z^x(\lambda_x)\ne Y_{f(x)}$, which has probability at most $\bar\eta$;
\item\label{itm:thm:intro-approx-random_2}
if $Z^{x\oplus J}(\lambda_x)$ lies on the wrong side of $\mathcal A$. This has probability at most $\bar\delta+\bar\eta$, where $\bar\delta$ comes from the definition of $\Acal$ and $\bar\eta$ comes from the distance between $Z^{x\oplus J}$ and the corresponding $\Pyes,\Pno$.
\end{enumerate}
Thus for each $x$ and $J$, we have
$$
\Pr\sbra{D_{J,f(x)}((L_i(x))_{i\in J})\ne f(x\oplus J)}\le\bar\delta+2\cdot\bar\eta=\theta<1/2,
$$
which means the average fraction of decoding errors over $R$, using $W$ and $D_{J,b}$'s, is at most $\theta$.
Fix all random choices so that this bound holds. We then need only record a correction pattern affecting at most a $\theta$ fraction of $R$; the number of such patterns is at most
$$
\binom{|R|}{\le\theta|R|}\le2^{\Hcal(\theta)\cdot|R|}
$$
possibilities.
As in the proof of
\Cref{thm:intro-random}, the total number of possible descriptions is at most
\[
  \pbra{2^{K^2}}^{2\binom n2}
  \cdot\pbra{K^{|W|}}^n
  \cdot2^{|W|}
  \cdot2^{\Hcal(\theta)|R|}
  =2^{|R|\cdot\beta},
\]
where
\begin{align*}
\beta
&=\Hcal(\theta)+\frac{|W|}{|R|}+\frac{n|W|\log K}{|R|}+\frac{2K^2\binom n2}{|R|}\\
&=\Hcal(\theta)+\frac1{\binom n2}+\frac{n\log K}{\binom n2}+\frac{2K^2}{|W|}
\tag{by \Cref{eq:thm:intro-approx-random_1}}\\
&=\Hcal(\theta)+2c+o(1).
\tag{since $K=\lceil2^{cn}\rceil$, $c<1/2$, and $|W|=\Omega\pbra{\frac{2^n}{n^4}}$}
\end{align*}
Since $\Hcal(\theta)+2c<1$, for $n\ge n_\star$ sufficiently large, we have $\beta<1$ and the fraction of truth tables on $R$ admitting such a description is therefore at most
$2^{-(1-\beta)\cdot|R|}\le2^{-\gamma\cdot2^n/n^2}$, where $\gamma$ is a constant depending on $\theta$ and $c$.

Since both $\gamma$ and $n_\star$ depend only on $\theta,c$, which in turn depend only on $\bar\delta,\bar\eta$, it suffices to define $c_{\bar\delta,\bar\eta}=\min\cbra{1/2,\gamma,c,1/n_\star}$ and this completes the proof.
\end{proof}

\subsection{Symmetric classification with statistical error}\label{sec:approx-classification}

We next prove \Cref{thm:approx-symmetric}, the robust version of \Cref{thm:intro-symmetric}.

\thmintroapproxsymmetric*

The proof follows the same roadmap as the perfect case, with robust substitutes for its two structural lemmas.

\begin{lemma}[Statistical version of \Cref{lem:block-identification}]\label{lem:approx-block-identification}
Assume $f\colon\bin^n\to\bin$ has a statistical DRE with alphabet $[K]$, correctness error $\delta$, and privacy error $\eta$.
The following transformations from \Cref{lem:block-identification} remain valid:
\begin{itemize}
\item The \textsc{Rearranging}, \textsc{Reverse}, \textsc{Negation}, and \textsc{Projection} rules preserve both error bounds.
\item The \textsc{Restriction} rule preserves both error bounds at the cost of a larger alphabet: fixing any $s<n$ input bits yields a statistical DRE with alphabet $[K^{s+1}]$.
\end{itemize}
\end{lemma}
\begin{proof}
The proof is almost identical to that of \Cref{lem:block-identification}.  For \textsc{Restriction}, append the encoding of the $s$ fixed input bits to an unrestricted input bit, which now has alphabet $[K]^{s+1}\cong[K^{s+1}]$. The new encoding is an injective rewriting of the original one on the restricted subcube.
\end{proof}

\begin{lemma}[Statistical version of \Cref{lem:local-labels}]\label{lem:approx-local-labels}
Fix an arbitrary $f\colon\bin^n\to\bin$ and a statistical DRE of $f$ with alphabet $[K]$, correctness error $\delta$, and privacy error $\eta$.
For any $X\subseteq\bin^n$ and $\Jcal\subseteq2^{[n]}$, if
$$
|X|\cdot(1+|\Jcal|)\cdot(\delta+\eta)<1,
$$
then there exist functions
\begin{itemize}
\item $D_{J,b}\colon[K]^J\to\bin$, where $J\in\Jcal$ and $b\in\bin$,
\item $L_i\colon\bin^n\to[K]$, where $i\in[n]$,
\end{itemize}
such that $f(x\oplus J)=D_{J,f(x)}\pbra{(L_i(x))_{i\in J}}$ holds for all $J\in\Jcal$ and $x\in X$.
\end{lemma}
\begin{proof}
The proof is essentially contained in the proof of \Cref{thm:intro-approx-random} above.
Using the constructions of $D_{J,b}$'s and $L_i$'s there, we have
\begin{align*}
&\phantom{\le}\Pr\sbra{\exists x\in X,J\in\Jcal,~D_{J,f(x)}((L_i(x))_{i\in J})\ne f(x\oplus J)}\\
&\le\Pr\sbra{\text{one of the events in \Cref{itm:thm:intro-approx-random_1,itm:thm:intro-approx-random_2} occurs for some $x\in X,J\in\Jcal$}}\\
&\le|X|\cdot\eta+|X|\cdot|\Jcal|\cdot(\delta+\eta)\\
&\le|X|\cdot(1+|\Jcal|)\cdot(\delta+\eta)<1.
\end{align*}
Hence there exists a choice that establishes the statement.
\end{proof}

We can now prove \Cref{thm:approx-symmetric}.

\begin{proof}[Proof of \Cref{thm:approx-symmetric}]
We only highlight the necessary changes to the proof of \Cref{thm:intro-symmetric}.

We first argue that $f$ must have period $a_K$ between Hamming weights $B_K$ and $n-B_K$, where $a_K,B_K$ are properly chosen as in \Cref{lem:periodic-core}.
This part uses the same proof strategy of gluing short periodic intervals together.
The only change lies in the proof of \Cref{lem:local-periodicity}: we will replace \Cref{lem:local-labels} with \Cref{lem:approx-local-labels} and need only ensure \Cref{eq:lem:local-periodicity_1} for all subsets $J$ of\footnote{We use $O_K(1)$ to denote any constant depending only on $K$.} $O_K(1)$ coordinates.
The key point is that the proof of \Cref{lem:periodic-core} only needs the local intervals in \Cref{lem:local-periodicity} to have length at least $a_K+1$, so they can overlap on an interval of length $a_K$ to transfer the periodic pattern.

The next step is to embed a large OR function, witnessing the mismatch from boundary Hamming weights.
We simply use \Cref{lem:approx-block-identification} in place of \Cref{lem:block-identification} for the reduction. The key observation is that, although the \textsc{Restriction} rule in \Cref{lem:approx-block-identification} now blows up the alphabet size, we only need to restrict $O_K(1)$ input bits, since we have already established periodicity between weights $B_K$ and $n-B_K$.

Finally we need a lower bound for OR: a statistical version of \Cref{thm:general-or}.
For this purpose, we simply use \Cref{lem:approx-local-labels} to mimic the proof of \Cref{lem:anti-embedding_family} and obtain a local anti-embedding family of dimension $O_K(1)$. By ensuring that this dimension is sufficiently large, the desired lower bound follows from \Cref{lem:nonexist_anti}.
\end{proof}

\section{Tight structured lower bound for OR}\label{sec:structured-or}

This section proves \Cref{thm:intro-permutation-closed}, providing evidence toward \Cref{conj:or-optimal} under a structural assumption.

Fix a perfect DRE $Z=(Z_i^b)_{i\in[n],b\in\bin}$ of $\OR_n$ with alphabet $[K]$. Let $\Pyes$ denote the distribution of $Z^x$ for any $x\ne0^n$, and let $\Pno$ denote the distribution of $Z^{0^n}$. We say that the DRE has symmetric support on input $0^n$ if, whenever $z\in[K]^n$ lies in the support of $\Pno$, every coordinate permutation of $z$ also lies in the support.

\thmintropermutationclosed*

We first record an elementary consequence of privacy.

\begin{lemma}\label{lem:proper-marginal}
If $I\subsetneq[n]$, then $(Z_i^{y_i})_{i\in I}$ has the same distribution for every $y\in\bin^I$.
\end{lemma}
\begin{proof}
Extend each $y$ to some $x\in\bin^n$ by setting the remaining bits to $1$. Hence each $Z^x$ has the same distribution $\Pyes$, which remains the same by taking the marginal on $I$.
\end{proof}

For $z\in[K]^*$ and $c\in[K]$, let $\hsf_c(z)$ be the number of occurrences of $c$ in $z$, and write $\hsf(z)=(\hsf_c(z))_{c\in[K]}$ for the histogram of $z$.

\begin{proof}[Proof of \Cref{thm:intro-permutation-closed}]
We will prove a quantitative bound $K=\Omega(n^{1/3})$, which translates to the size lower bound $n\log K=\Omega(n\log n)$.

We may assume $n$ is even. Indeed, if $n$ is odd, fix the last input to zero and condition on any positive-probability value $q$ of its encoding, as in the restriction proof of \Cref{lem:block-identification}. Deleting that coordinate gives a perfect DRE of $\OR_{n-1}$ over the same alphabet. Its zero-input support is symmetric: any permutation of the remaining coordinates extends to a permutation fixing the last coordinate, so it preserves both the original zero-input support and the conditioning event. Replacing $n$ by $n-1$ does not affect the asymptotic bound.
Choose a uniformly random $S\subset[n]$ of size $n/2$ and write
$\bar S=[n]\setminus S$.
For each $c\in[K]$, let
\[
  t_c=\hsf_c((Z_i^1)_{i\in[n]}),
  \quad\text{and}\quad
  D_c=\hsf_c((Z_i^1)_{i\in S})-\hsf_c((Z_i^1)_{i\in\bar S}).
\]
A simple calculation shows, for fixed $Z$ and random $S$, we have
\[
  \E[D_c]=0,
  \qquad
  \operatorname{Var}(D_c)
  =\frac{t_c(n-t_c)}{n-1}\le t_c,
\]
which means $\E_S|D_c|\le\sqrt{t_c}$ and
$$
\E_S\vabs{\hsf((Z_i^1)_{i\in S})-\hsf((Z_i^1)_{i\in\bar S})}_1=\sum_{c\in[K]}\E_S|D_c|\le\sum_{c\in[K]}\sqrt{t_c}\le\sqrt{K\sum_c t_c}=\sqrt{Kn}.
$$
Hence
$$
\E_{Z,S}\vabs{\hsf((Z_i^1)_{i\in S})-\hsf((Z_i^1)_{i\in\bar S})}_1\le\sqrt{Kn}.
$$

For every fixed $S$, we know that $Z^{e_{\bar S}}$ and $Z^{1^n}$ have the same distribution $\Pyes$.
Hence the above analysis also shows
$$
\E_{Z,S}\vabs{\hsf((Z_i^0)_{i\in S})-\hsf((Z_i^1)_{i\in\bar S})}_1\le\sqrt{Kn}.
$$
By the triangle inequality, we have
\begin{equation}\label{eq:AB-histogram}
\E_{Z,S}\vabs{\hsf((Z_i^0)_{i\in S})-\hsf((Z_i^1)_{i\in S})}_1\le2\sqrt{Kn}.
\end{equation}

Fix $Z$ and $S$ for which the norm inside the expectation in \Cref{eq:AB-histogram} is at most $2\sqrt{Kn}$.
Construct a directed multigraph $G$ on $[K]$ with $n/2$ edges: for each $i\in S$, we add an edge from $Z_i^0$ to $Z_i^1$.
By \Cref{lem:pointwise-separation}, $G$ has no self-loops.
Since our DRE has symmetric support on input $0^n$, $G$ cannot have a directed cycle.
To see this, a directed cycle indexed by $T\subseteq S$ has equal source and target multisets, which means $Z^{e_T}$ is equivalent to $Z^{0^n}$ up to permutation.
The assumed symmetry also puts $Z^{e_T}$ in the support of $\Pno$, which contradicts the correctness property in \Cref{def:intro-dre}.

The acyclicity of $G$ defines a topological order $\rho\colon[K]\longrightarrow\{0,1,\ldots,K-1\}$ that strictly increases along every edge.
Hence
$$
  \frac n2
  \le\sum_{i\in S}\pbra{\rho(Z_i^1)-\rho(Z_i^0)}.
$$
On the other hand, since $0\le\rho(c)\le K-1$ for every $c\in[K]$,
\begin{align*}
\sum_{i\in S}(\rho(Z_i^1)-\rho(Z_i^0))
&=\sum_{c\in[K]}\rho(c)\cdot\pbra{\hsf_c((Z_i^1)_{i\in S})-\hsf_c((Z_i^0)_{i\in S})}\\
&\le K\cdot\vabs{\hsf((Z_i^1)_{i\in S})-\hsf((Z_i^0)_{i\in S})}_1\\
&\le2K^{3/2}\sqrt n.
\tag{by \Cref{eq:AB-histogram}}
\end{align*}
This gives $K=\Omega(n^{1/3})$ by rearranging.
\end{proof}

\bibliographystyle{alphaurl}
\bibliography{ref}

\end{document}